\documentclass[letterpaper, 10 pt, conference]{ieeeconf}  

\IEEEoverridecommandlockouts                              

\usepackage{amsfonts}
\usepackage{amsmath}
\usepackage{ifthen}
\usepackage{fancyhdr}
\usepackage{amsfonts}
\usepackage{graphicx}
\usepackage{amsmath}
\usepackage{ifthen}
\usepackage{color}
\usepackage{indentfirst}
\usepackage{epstopdf}
\usepackage{cite}

\DeclareMathOperator{\sgn}{sgn}
\newcommand{\E}{\mathbb{E}}

\newcommand{\Prob}{\mathbb{P}}
\makeatletter
\newcommand*{\rom}[1]{\expandafter\@slowromancap\romannumeral #1@}
\makeatother
\newtheorem{theorem}{Theorem}

\newtheorem{remark}{Remark}

\title{\LARGE \bf Optimal Estimation with Limited Measurements and Noisy Communication}
\author{Xiaobin Gao, Emrah Akyol, and Tamer Ba\c{s}ar \thanks{This research was supported in part by NSF under grant CCF 11-11342, and in part by the U.S. Air Force Office of
Scientific Research (AFOSR) MURI grant FA9550-10-1-0573.}\thanks{All authors are with the Coordinated Science Laboratory, University of Illinois at Urbana-Champaign, Urbana, IL 61801; emails: \{xgao16, akyol, basar1\}@illinois.edu}}

\begin{document}

\maketitle
\thispagestyle{empty}
\pagestyle{empty}

\begin{abstract}
This paper considers a sequential estimation and sensor scheduling problem with one sensor and one estimator.  The sensor makes sequential observations about the state of an underlying memoryless stochastic process, and makes a decision as to whether or not to send this measurement to the estimator.  The sensor and the estimator have the common objective of minimizing expected distortion in the estimation of the state of the process, over a finite time horizon, with the constraint that the sensor can transmit its observation only a limited number of times.  As opposed to the prior work where communication between the sensor and the estimator was assumed to be perfect (noiseless), in this work an additive noise channel with fixed power constraint is considered; hence, the sensor has to encode its message before transmission.  For some specific source and channel noise densities, we obtain the optimal encoding and estimation policies in conjunction with the optimal  transmission schedule. The impact of the presence of a noisy channel is analyzed numerically based on dynamic programming. This analysis yields some rather surprising results such as a phase-transition phenomenon in the number of used transmission opportunities, which was not encountered in the noiseless communication setting.

\end{abstract}

\section{Introduction}
Joint sensor scheduling and remote state estimation problems have recently gained renewed interest due to proliferation of energy limited sensor networks, see e.g., \cite{athans1972determination,mehra1976optimization, Hespanha07,Imer10, Lipsa11,Nayyar13,Wu13} and the references therein.

In \cite{Imer10}, the following problem was considered: Estimate a one-dimensional discrete-time stochastic process distributed independently and identically (i.i.d.) over a decision horizon of length $T$ using only $N \leq T$ measurements. Both the measurement and the estimation of the process were carried out sequentially by two different decision makers, the sensor and the estimator. Over the decision horizon of length $T$, the sensor had exactly $N$ opportunities to transmit its observation to the estimator. These transmissions were assumed to be {\it error and noise free}, and the problem posed was to jointly determine the best sensing and estimation policies that minimize the average estimation error between the  process and its estimate.
Optimum transmission decisions were sought in the class of threshold based strategies and the optimal decision sequence, i.e., the evolution of the thresholds in time based on the realization of the process, were obtained via dynamic programming. Later, using majorization and related techniques, such threshold based strategies were shown to be optimal for this problem \cite{Lipsa11,Nayyar13} and even for more general settings where the process is not necessarily memoryless \cite{Lipsa11,Nayyar13}.

Note that all prior work considered the problem with perfect (noiseless) communication between the sensor and the estimator, which was an important starting point for this line of research. More realistic scenarios, however, are those where the transmission channels are noisy---a problem that presents several challenges. The main difficulty here is that with noise in the channel, and under an average power constraint, the sensor has to encode its message before transmission, and the estimator has to consider this encoding mapping in its estimation mapping. However, the optimal zero-delay encoding/estimation mappings are not known in general, except in the Gaussian source-Gaussian channel case; see e.g., \cite{CoverBook}, for which the mappings are known to be linear (or affine if the random variables are not zero-mean) for all power levels. Even in this special case, however, once the sensor observation is thresholded, the distribution is no longer Gaussian, and hence linear (or affine) mappings may no longer be optimal, making the problem fairly intractable.

In \cite{Akyol13}, the settings where linear (or affine) strategies are optimal for zero-delay communication have been characterized in terms of the source and the additive channel noise distributions. It was shown that if and only if a ``matching condition", defined over the characteristic functions of the source and the channel noise, is satisfied, then the linear encoding/estimation policies are optimal. This characterization enables tractability of the zero-delay communication problems, beyond the Gaussian source-Gaussian channel case. Implications of this matching condition on the adversarial zero-delay communication was studied in \cite{akyol2013optimal}, where it was shown that the optimal strategy for an adversarial agent with fixed jamming power is to render the effective channel noise distribution to match that of the source, so that the matching conditions are satisfied, and the optimal encoding/decoding mappings are linear.

In \cite{Gao15}, we applied the matching condition of \cite{Akyol13} to the problem of sensor scheduling and remote estimation, and showed optimality of threshold based sensor scheduling policies and affine encoding/estimation policies for the case of Laplacian source and Gamma channel noise, with {\it soft constraints} over the number of transmissions. In this paper, building on this recent prior work \cite{Gao15}, we consider a hard constraint on the number of transmission times over a finite time horizon, i.e, we extend the work in \cite{Imer10} to noisy communication settings for Laplacian source and Gamma channel noise. Using a dynamic programming approach, as done in \cite{Imer10}, we obtain the optimal transmission scheduling policy.  Beyond the expected results, we notice some rather surprising effects of the noisy communication considerations in this class of remote estimation problems. For example, over a time horizon $T$ and with a hard transmission limit, $N\leq T$, if the state realizations were so that at time step $K$, the sensor has used only $N-T+K$ transmissions out of $N$, the intuitively appealing solution to the noiseless variation of the problem was to transmit all the observed state realizations without any thresholding, i.e., the threshold is effectively set to zero for samples at time steps $K+1, \dots, T$. However, in the noisy setting, we have noticed that this is not the case, the sensor may not use all the transmission opportunities left. This is due to the fact that threshold information--that is whether or not the state sample belongs to an interval-- may be  more valuable than a ``noisy" observation of the state. In fact, depending on the signal-to-noise ratio (SNR) of the channel, there is a fixed number of useful (in average) number of transmissions, and allowing transmissions more than this number, on the average, does not help decrease the expected mean square error (MSE). 

The rest of the paper is organized as follows. In Section \ref{ProblemFormulation}, we formulate the problem. In Section \ref{Prior Work}, we present some preliminary results. In Section \ref{Main result}, we present and prove the main results. In Section \ref{numerical results}, we present and discuss some numerical results. Finally in Section \ref{Conclusions}, we include conclusions and discuss some future directions.

\section{Problem Formulation}
\label{ProblemFormulation}
\subsection{System Model}
\begin{figure}[h]
\centering
\includegraphics[height=24mm]{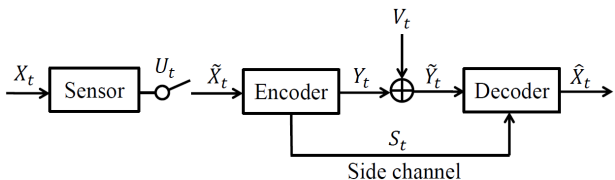}
\caption{System model}
\end{figure}
Consider a remote sensing and estimation system propagating in discrete time over a finite time horizon, namely, $t=1,2,\ldots,T$. In the system, there is \emph{one} remote sensor, \emph{one} encoder and \emph{one} estimator (which is also called ``decoder"). The sensor takes measurements on a one-dimensional, independent identically distributed (i.i.d.) random process $\{X_t\}$, which has Laplace distribution with parameters $(0,\lambda^{-1})$.
Assume that at time $t$, the sensor takes a perfect measurement on $X_t$, and then decides whether to transmit the measurement to the encoder or not. Let a binary variable $U_t$ be the sensor's decision at time $t$, where $U_t=0$ stands for no transmission and $U_t=1$ stands for transmission. The sensor is restricted to make no more than $N$ times of transmissions over the time horizon, that is
$$
\sum_{t=1}^T U_t \leq N
$$
Let $\tilde{X}_t$ be the message received by the encoder. Assume that the transmission from the sensor to the encoder is perfect, then
\begin{equation*}
\tilde{X}_t =
\begin{cases}
X_t, & \mbox{if } U_t=1 \vspace{0.2cm}
\\ \epsilon, & \mbox{if } U_t=0
\end{cases}
\end{equation*}
where $\epsilon$ is a free symbol standing for no transmission is made. After receiving the message from the sensor, the encoder sends an encoded message to the communication channel, denoted by $Y_t$, $Y_t\in \mathbb{R}$. The encoder is not able to send any message to the communication channel if it does not receive any message from the sensor, which is denoted by $Y_t=\epsilon$. The encoder has average power constraint:
$$
\mathbb{E}[Y_t^2|U_t=1]\leq P_T
$$
where $P_T$ is known, and this holds for all $t$. Assume that the encoded message is disturbed by an additive channel noise $V_t$. $\{V_t\}$ is an i.i.d. random process with Gamma distribution $\Gamma(k,\theta)$, which is independent of $\{X_t\}$. Let $\tilde{Y}_t$ be the noise-corrupted message received by the decoder, we have
\begin{equation*}
\tilde{Y}_t =
\begin{cases}
Y_t+V_t, & \mbox{if } Y_t\neq\epsilon \vspace{0.2cm}
\\ \epsilon, & \mbox{if } Y_t=\epsilon
\end{cases}
\end{equation*}
When sending the encoded message $Y_t$, the encoder is able to transmit the sign of $\tilde{X}_t$ to the decoder via a noiseless side channel, denoted by $S_t$. Again, the encoder is not able to send any message to the decoder via the side channel if it does not receive any message from the sensor, then
\begin{equation*}
S_t =
\begin{cases}
\sgn(\tilde X_t), & \mbox{if } \tilde X_t\neq\epsilon \vspace{0.2cm}
\\ \epsilon, & \mbox{if } \tilde X_t=\epsilon
\end{cases}
\end{equation*}
Based on the received messages $\tilde{Y}_t$ and $S_t$, the decoder generates an estimate on $X_t$, denoted by $\hat{X}_t$. The decoder is charged for distortion in estimation. Assume that the distortion function $\rho(X_t,\hat{X}_t)$ is the squared error $(X_t-\hat{X}_t )^2$, and the cumulative distortion is the sum of stage-wise squared errors over the decision horizon.

\subsection{Decision Strategies}
Assume that at time $t$, the sensor has memory on all its measurements by $t$, denoted by $X_{1:t}$, and all the decisions it has made by $t-1$, denoted by $U_{1:t-1}$. The sensor makes decision $U_t$ based on its current information $(X_{1:t},U_{1:t-1})$, that is
$$
U_t=f_t(X_{1:t},U_{1:t-1})
$$
where $f_t$ is the sensor scheduling policy at time $t$ and $\textbf{f}=\{f_1,f_2,\ldots,f_T\}$ is the sensor scheduling strategy.

Similarly, at time $t$, the encoder is assumed to have memory on all the messages received from the sensor by $t$, denoted by $\tilde{X}_{1:t}$, and all the encoded messages it has sent to the communication channel by $t-1$, denoted by $Y_{1:t-1}$. The encoder generates the encoded message $Y_t$ based on its current information $(\tilde{X}_{1:t},Y_{1:t-1})$, that is
$$
Y_t=g_t(\tilde{X}_{1:t},Y_{1:t-1})
$$
where $g_t$ is the encoding policy at time $t$ and $\textbf{g}=\{g_1,g_2,\ldots,g_T\}$ is the encoding strategy.

Finally, it is assumed that at time $t$, the decoder has memory on all the messages received from the encoder by $t$, denoted by $\tilde{Y}_{1:t},S_{1:t}$. The decoder produces estimate $\hat{X}_t$ based on its current information $(\tilde{Y}_{1:t},S_{1:t})$, namely
$$
\hat{X}_t=h_t(\tilde{Y}_{1:t},S_{1:t})
$$
where $h_t$ is the decoding policy at time $t$ and $\textbf{h}=\{h_1,h_2,\ldots,h_T\}$ is the decoding strategy.
\begin{remark} \label{Shared information} Although we do not assume that the encoder and the decoder have memory on $U_{1:t}$, yet they can deduce $U_{1:t}$ from $\tilde{X}_{1:t}$ and $\tilde{Y}_{1:t}$, respectively. Similarly, the decoder can deduce the previous estimates $\hat{X}_{1:t-1}$ from $(\tilde{Y}_{1:t-1},S_{1:t-1})$ and $\{h_1,h_2,\ldots,h_{t-1}\}$. \end{remark}

\subsection{Assumptions on the Parameters}
Let $\sigma_V^2$ be the variance of $V_t$. Since $V_t$ has gamma distribution $\Gamma(k,\theta)$, $\sigma_V^2=k\theta^2$. Let $\alpha:=\lambda\sqrt{P_T}$, and $\gamma:=\frac{P_T}{\sigma_V^2}$. $\gamma$ is also called signal to noise ratio (SNR). Assume that
$$
\theta=\sqrt{P_T}
$$
Then, we have
\begin{equation}
\label{assumptions on parameters}
\alpha=\lambda\theta \mbox{, } \gamma=\frac{1}{k}
\end{equation}

A detailed explanation of the motivation for these assumptions can be found in \cite[Remark 2]{Gao15}.



\subsection{Optimization Problem}
Consider the system described above, given the time horizon $T$, the number of transmission opportunities $N$, the statistics of $\{X_t\}$ and $\{V_t\}$, and the power constraint $P_T$. Determine the sensor scheduling strategy, encoding strategy, and decoding strategy $(\textbf{f},\textbf{g},\textbf{h})$ that minimize the expected value of the sum of stage-wise estimation costs over the time horizon, that is,
$$
J(\textbf{f},\textbf{g},\textbf{h})= \mathbb{E} \left\{\sum_{t=1}^T {(X_t-\hat{X}_t)}^2\right\}
$$
subject to the communication constraint of the sensor and the power constraint of the encoder.

\section{Prior Work}
\label{Prior Work}
Consider the sensor scheduling and remote estimation problem described above, but with the following modifications:
\begin{enumerate}
    \item The time horizon $T=1$ (and hence we suppress in this section the subscript for time in all the expressions).
    \item The sensor is not constrained by the number of transmissions. Instead, it is charged a cost $c$ if it transmits its observation. No transmission means no communication cost.
    \item The optimization problem is to design the scheduling policy, encoding policy, and decoding policy $(f,g,h)$ that minimize the following cost function:
$$
J(f,g,h)= \mathbb{E} \left\{cU + {(X-\hat{X})}^2\right\}
$$
\end{enumerate}


\begin{theorem}\cite{Gao15} Consider the communication problem described above, and restrict the sensor to apply symmetric threshold based policy, that is
$$
U=f(X)=
\begin{cases}
1, & \mbox{if } |X|>\beta \vspace{0.2cm}
\\ 0, & \mbox{if } |X|\leq\beta
\end{cases}
$$
where $\beta>0$ is the threshold. Then,
\begin{enumerate}
    \item The optimal scheduling policy is the one with threshold $\beta^\ast = \sqrt{c+m}$, where $m=\frac{1}{\gamma+1}\frac{1}{\lambda^2}$
    \item The optimal encoding and decoding policies are as follows
\end{enumerate}
\label{One-stage Problem}
$$
\begin{array}{lcr}
g(\tilde{X})\;\;\;\;=
\begin{cases}
\alpha|\tilde{X}|-\alpha\beta^\ast-\alpha\lambda^{-1}, & \mbox{if } \tilde{X} \neq \epsilon \vspace{0.2cm}
\\ \epsilon, & \mbox{if } \tilde{X} = \epsilon
\end{cases}
\vspace{0.2cm}
\\h(\tilde{Y},S)=
\begin{cases}
S\cdot\left(\frac{1}{\alpha}\frac{\gamma}{\gamma+1}\tilde{Y}+\frac{\gamma}{\gamma+1}\lambda^{-1}+\beta^\ast \right),  \mbox{if } \tilde{Y},S \neq \epsilon \vspace{0.2cm}
\\ 0, \;\;\;\;\;\;\;\;\;\;\;\;\;\;\;\;\;\;\;\;\;\;\;\;\;\;\;\;\;\;\;\;\;\;\;\;\;\;\;\;\;\;\;\;\; \mbox{if } \tilde{Y},S = \epsilon
\end{cases}
\end{array}
$$
where $\alpha = \lambda\sqrt{P_T}$, $\gamma=\frac{P_T}{k\theta^2}$.
\end{theorem}

\section{Main Results}
\label{Main result}
We first define $E_t$ as the number of communication opportunities left at time $t$, i.e.,
$$
E_t = N-\sum_{i=1}^{t-1} U_i
$$
Then, the communication constraint can be expressed by
$$
U_t \leq E_t, \;\;\; \forall \; t = 1,2,\ldots,T
$$
By Remark \ref{Shared information}, $U_{1:t-1}$ is the common information shared by the sensor, the encoder, and the decoder, and hence $E_t$ is also known by all the decision makers. Then we have the following theorem.
\begin{theorem}
\label{consider only on current states}
Consider the sensor scheduling and remote estimation problem described in section \ref{ProblemFormulation}.
Without loss of optimality, we can restrict the sensor scheduling, encoding and decoding policies to the forms:
$$
U_t = f_t (X_t,E_t) \mbox{, } Y_t = g_t (\tilde{X}_t,E_t) \mbox{, } \hat{X}_t = h_t(\tilde{Y}_t,S_t,E_t)
$$
\end{theorem}
\begin{proof}At time $t=T$, we want to design $(f_T,g_T,h_T)$ to minimize
$$
J_{T_1}(f_T,g_T,h_T)=\E\left\{{(X_T-\hat{X}_T)}^2\right\}
$$
where $U_T=f_T(X_{1:T},U_{1:T-1})$, $Y_T=g_T(\tilde{X}_{1:T},Y_{1:T-1})$, $\hat{X}_T$ $=h_T(\tilde{Y}_{1:T},S_{1:T})$. We call this problem \textit{Problem T1}. Denote by $I_{sT},I_{eT},I_{dT}$ the information about the past system states available to the sensor, the encoder, and the decoder, respectively, at time $T$, i.e., $I_{sT} = \{X_{1:T-1},U_{1:T-1}\}$, $I_{eT} = \{\tilde{X}_{1:T-1},Y_{1:T-1}\}$, and $I_{dT} = \{\tilde{Y}_{1:T-1},S_{1:T-1}\}$. Then the decisions at time $T$ are generated by $U_T=f_T(X_T,I_{sT})$, $Y_T=g_T(\tilde{X}_T,I_{eT})$, $\hat{X}_T=h_T(\tilde{Y}_T,S_T,I_{dT})$.

Denote by $I_T$ the information set about the past system states at time $T$, namely,
$$
I_T=\{X_{1:T-1},U_{1:T-1},\tilde{X}_{1:T-1},Y_{1:T-1},\tilde{Y}_{1:T-1},S_{1:T-1}\}
$$
Then $I_{sT},I_{eT},I_{dT}\in I_T$. Consider another problem, which we call \textit{Problem T2}, where $I_T$ is available to the sensor, the encoder, and the decoder, and we want to design $(f_T^{\prime},g_T^{\prime},h_T^{\prime})$ to minimize
$$
J_{T_2}(f_T^\prime,g_T^\prime,h_T^\prime)=\E\left\{{(X_T-\hat{X}_T)}^2\right\}
$$
where $U_T=f_T^\prime(X_T,I_T)$, $Y_T=g_T^\prime(\tilde{X}_T,I_T)$, $\hat{X}_T=  h_T^\prime(\tilde{Y}_T,S_T,I_T)$. Since the sensor, the encoder, and the decoder can always ignore the redundant information and behave as if they only know $I_{sT},I_{eT},I_{dT}$, respectively, the system in \textit{Problem T2} cannot perform worse than the system in \textit{Problem T1}, i.e.,
$$
\underset{(f_T^\prime,g_T^\prime,h_T^\prime)}{\min}J_{T_2}(f_T^\prime,g_T^\prime,h_T^\prime) \leq \underset{(f_T,g_T,h_T)}{\min}J_{T_1}(f_T,g_T,h_T)
$$
Similarly, consider a third problem, which we call \textit{Problem T3}, where only $E_T$ is available to the sensor, the encoder, and the decoder. We want to design $(f_T^{\prime\prime},g_T^{\prime\prime},h_T^{\prime\prime})$ to minimize
$$
J_{T_3}(f_T^{\prime\prime},g_T^{\prime\prime},h_T^{\prime\prime})=\E\left\{{(X_T-\hat{X}_T)}^2\right\}
$$
where $U_T=f_T^{\prime\prime}(X_T,E_T)$, $Y_T=g_T^{\prime\prime}(\tilde{X}_T,E_T)$, $\hat{X}_T$ $= h_T^{\prime\prime}(\tilde{Y}_T,S_T,E_T)$. Since $E_T$ can be deduced from $I_{sT},I_{eT},I_{dT}$, by a similar argument as above, the performance of the system in \textit{Problem T1} is no worse than the performance of the system in \textit{Problem T3}, that is,
$$
\underset{(f_T,g_T,h_T)}{\min}J_{T_1}(f_T,g_T,h_T) \leq \underset{(f_T^{\prime\prime},g_T^{\prime\prime},h_T^{\prime\prime})}{\min}J_{T_3}(f_T^{\prime\prime},g_T^{\prime\prime},h_T^{\prime\prime})
$$
Let us now return to \textit{Problem T2}. Since the distortion function $\rho(\cdot,\cdot)$ and the power constraint of the encoder do not depend on $I_T$, the communication constraint depends on $I_T$ only via $E_T$,
and $\{X_t\}$ and $\{V_t\}$ are i.i.d. random processes, $X_T$ and $V_T$ are also independent of $I_T$ and there is no loss of optimality if we restrict $U_T = f_T^\prime(X_T,E_T) \mbox{, } Y_T = g_T^\prime(\tilde{X}_T,E_T) \mbox{, }  \hat{X}_T = h_T^\prime(\tilde{Y}_T,S_T,E_T)$
and
$$
\underset{(f_T^\prime,g_T^\prime,h_T^\prime)}{\min}J_{T_2}(f_T^\prime,g_T^\prime,h_T^\prime) = \underset{(f_T^{\prime\prime},g_T^{\prime\prime},h_T^{\prime\prime})}{\min}J_{T_3}(f_T^{\prime\prime},g_T^{\prime\prime},h_T^{\prime\prime})
$$
The equality above shows that in \textit{Problem T1} the sensor, the encoder and the decoder can ignore their information about the past, namely $I_{sT}$, $I_{eT}$, and $I_{dT}$, respectively, but just consider $E_T$.
Moreover, the optimal cost at time $T$ is a function of $E_T$, denoted by $J^\ast(T,E_T)$. Note that the evolution of $E_t$ is described by
\begin{equation}
\label{evolution of E_t}
\begin{array}{lcl}
E_1 &=& N \vspace{0.2 cm}
\\ E_t &=& E_{t-1} - U_{t-1} \mbox{, }\;\;\;t\geq2
\end{array}
\end{equation}
Therefore at time $T-1$ we want to design $(f_{T-1}$, $g_{T-1}$ , $h_{T-1})$ to minimize
$$
\begin{array}{lll}
\;\;\;J_{T-1}(f_{T-1},g_{T-1},h_{T-1}) \vspace{0.2 cm}
\\ = \E\left\{{(X_{T-1}-\hat{X}_{T-1})}^2\right\} + \E\left\{{(X_T-\hat{X}_T)}^2\right\} \vspace{0.2 cm}
\\ = \E\left\{{(X_{T-1}-\hat{X}_{T-1})}^2\right\} + \E\Big\{J^\ast\big(T,E_{T})\big)\Big\} \vspace{0.2 cm}
\\ = \E\left\{{(X_{T-1}-\hat{X}_{T-1})}^2\right\}  + \E\Big\{J^\ast\big(T,E_{T-1} - U_{T-1}\big)\Big\}
\end{array}
$$
The first term depends on the choice of $(f_{T-1}$, $g_{T-1}$ , $h_{T-1})$, and the second part depends only on $E_{t-1}$ and $f_{T-1}$. By an argument similar to the one above, when minimizing $J_{T-1}(f_{T-1},g_{T-1},h_{T-1})$, it is sufficient for decision makers to consider only $E_{t-1}$ instead of $I_{s(T-1)},I_{e(T-1)},I_{d(T-1)}$.
Similarly, the optimal cost starting from time $T-1$ is a function of $E_{T-1}$, to be denoted by $J^\ast(T-1,E_{T-1})$. By induction one can show that there is no loss of generality by restricting $U_t = f_t (X_t,E_t) , Y_t = g_t (\tilde{X}_t,E_t) , \hat{X}_t = h_t(\tilde{Y}_t,S_t,E_t)$, and the optimal cost starting from time $t$ is a function of $E_{t}$, to be denoted by $J^\ast(t,E_t)$
\end{proof}

The proof of Theorem \ref{consider only on current states} also shows that the optimal cost function $J^\ast(t,E_t)$ and the optimal policies $(f^\ast_t,g^\ast_t,h^\ast_t)$ can be computed by the standard dynamic programming equation \cite{bertsekas1995dynamic} as follows,
$$
\begin{array}{ccl}
J^\ast(T+1,\cdot) = 0 \vspace{0.2 cm}
\\ J^\ast(t,E_t) = \underset{f_t,g_t,h_t}{\min} \Big\{\E[(X_t-\hat{X}_t)^2]+ \E[J^\ast(t+1,E_{t+1})]\Big\}
\end{array}
$$
where the evolution of $E_t$ is described by \eqref{evolution of E_t}, and $f_t( \cdot ,0)=0$ due to the constraint on the communication opportunities. Depending on the realization of $X_t$, $E_{t+1}$ may be $E_t$ or $E_t-1$. Therefore the dynamic programming equation can also be written as
\begin{equation}
\label{dynamic programming equation}
\begin{array}{lcl}
 J^\ast(t,E_t) = \underset{f_t,g_t,h_t}{\min} \bigg\{\E[(X_t-\hat{X}_t)^2]+J^\ast(t+1,E_t-1) \vspace{0.2 cm}
\\ \cdot \displaystyle \int f_t(x,E_t)p_X(x)dx + J^\ast(t+1,E_t)\cdot \displaystyle \int (1-f_t(x,E_t)) \vspace{0.2 cm}
\\ p_X(x)dx\bigg\} = J^\ast(t+1,E_t) + \underset{f_t,g_t,h_t}{\min} \bigg\{\E[(X_t-\hat{X}_t)^2] \vspace{0.2 cm}
\\ + c_t(E_t) \cdot \displaystyle \int f_t(x,E_t)p_X(x)dx \bigg\}
\end{array}
\end{equation}
where $c_t(E_t) = J^\ast(t+1,E_t-1)-J^\ast(t+1,E_t)$. Note that the minimization in \eqref{dynamic programming equation} above is just the one-stage problem discussed in section \ref{Prior Work} with communication cost $c_t(E_t)$. Hence we have the following theorem.
\begin{theorem}
\label{Optimal scheduling encoding decoding}
Consider the sensor scheduling and remote estimation problem described in section \ref{ProblemFormulation}. If we restrict the sensor to apply the symmetric threshold based policy (introduced in section \ref{Prior Work}), then the optimal policies for the sensor, the encoder, and the decoder can be described, respectively, as follows:
$$
\begin{array}{lcr}
f_t^\ast(X_t,E_t) \;\;\;\;=
\begin{cases}
1, & \mbox{if } E_t>0  \mbox{ and } |X_t|>\beta_t^\ast(E_t) \vspace{0.2cm}
\\ 0, & \mbox{if } E_t=0  \mbox{ or}\;\;\; |X_t|\leq\beta_t^\ast(E_t)
\end{cases}
\vspace{0.2cm}
\\g_t^\ast(\tilde{X}_t,E_t)\;\;\;\;=
\begin{cases}
\alpha|\tilde{X}_t|-\alpha\beta_t^\ast(E_t)-\alpha\lambda^{-1}, & \mbox{if } \tilde{X}_t \neq \epsilon \vspace{0.2cm}
\\ \epsilon, & \mbox{if } \tilde{X}_t = \epsilon
\end{cases}
\vspace{0.2cm}
\\h_t^\ast(\tilde{Y}_t,S_t,E_t)=
\begin{cases}
S_t\cdot\left(\frac{1}{\alpha}\frac{\gamma}{\gamma+1}\tilde{Y}_t+\frac{\gamma}{\gamma+1}\lambda^{-1}+\beta^\ast_t(E_t) \right), \vspace{0.2cm}
\\\;\;\;\;\;\;\;\;\;\;\;\;\;\;\;\;\;\;\;\;\;\;\;\;\;\;\;\;\;\;\;\;\;\;\;\;\;\; \mbox{if } \tilde{Y}_t,S_t \neq \epsilon \vspace{0.2cm}
\\ 0, \;\;\;\;\;\;\;\;\;\;\;\;\;\;\;\;\;\;\;\;\;\;\;\;\;\;\;\;\;\;\;\;\;\; \mbox{if } \tilde{Y}_t,S_t = \epsilon
\end{cases}
\end{array}
$$
where $\beta^\ast_t(E_t) = \sqrt{c_t(E_t)+m}$, $c_t(E_t) = J^\ast(t+1,E_t-1)-J^\ast(t+1,E_t)$, $m=\frac{1}{\gamma+1}\frac{1}{\lambda^2}$, $\alpha = \lambda\sqrt{P_T}$, $\gamma=\frac{P_T}{k\theta^2}$.
\end{theorem}

\begin{remark}
$c_t(E_t)$ can be interpreted as the opportunity cost for choosing to communicate with the estimator rather than not to communicate.
\end{remark}
\begin{remark}
Consider the case where $E_t > T-t$, that is, the sensor is always allowed to communicate with the estimator for the rest of time. First, we note that the opportunity cost $c_t(E_t)$ is zero. Also, even though the sensor can always communicate with the estimator, the optimal communication policy is still the threshold-based policy with threshold $\beta_t^\ast(E_t) = \sqrt{m} > 0$, which might seem counter-intuitive: why would the sensor not transmit its observation although it is allowed to do so?  This surprising result is due to the fact that threshold information, i.e., whether or not the state sample belongs to a fixed, known interval, might be more informative than a noisy observation of the state at the output of the noisy channel. Hence, it might be better not to communicate explicitly over the noisy channel but rely on the side channel which signals where the sample lies. For example, at the extreme case of a very noisy channel ($\gamma \rightarrow 0$) the output of the communication channel, $\tilde{Y}_t$, is effectively useless, irrespective of the realization $X_t$. However, depending on the threshold and the realization $X_t$, thresholding information could be significantly more informative.  \end{remark}

\section{Numerical Results}
\label{numerical results}
By plugging the optimal sensor scheduling, encoding, and decoding policies $(f^\ast_t,g^\ast_t,h^\ast_t)$ described in Theorem \ref{Optimal scheduling encoding decoding} into the dynamic programming equation \eqref{dynamic programming equation}, we get the explicit update rule for the optimal cost function $J^\ast(t,E_t)$, as shown below
$$
\begin{array}{lcl}
J^\ast(t,E_t) &=& J^\ast(t+1,E_t) + 2\lambda^{-2}, \;\;\;\;\;\;\;\;\;\;\;\; \mbox{ if }  E_t=0 \vspace{0.2 cm}
\\ J^\ast(t,E_t) &=& J^\ast(t+1,E_t) + 2\lambda^{-2}- 2\lambda^{-2}  \vspace{0.2 cm}
\\ &\;&  \cdot\big(2\beta^\ast_t(E_t)\lambda+ 1\big)\cdot e^{-\lambda\beta^\ast_t(E_t)},\;\; \mbox{ if }  E_t>0
\end{array}
$$
where $\beta^\ast_t(E_t) = \sqrt{c_t(E_t)+m}$, $m=\frac{1}{\gamma+1}\frac{1}{\lambda^2}$. The computation complexity of the dynamic programming equation is $O(TN)$, where $T$ is the time horizon and $N$ is the number of transmission opportunities.

We choose parameters as follows: $T=100$, $\lambda=1$. In particular, we choose $k = 10,1,0.1$, which corresponds to signal to noise ratio (SNR) $\gamma = 0.1,1,10$ by \eqref{assumptions on parameters}. We solve \eqref{dynamic programming equation} by applying the update rule for optimal cost function as described above. We plot the optimal 100-stage estimation error versus the number of communication opportunities under different SNRs, as shown in Figure \ref{fig: 100 stage versus numb of comm}.
\begin{figure}[h]
\centering
\includegraphics[height=60mm]{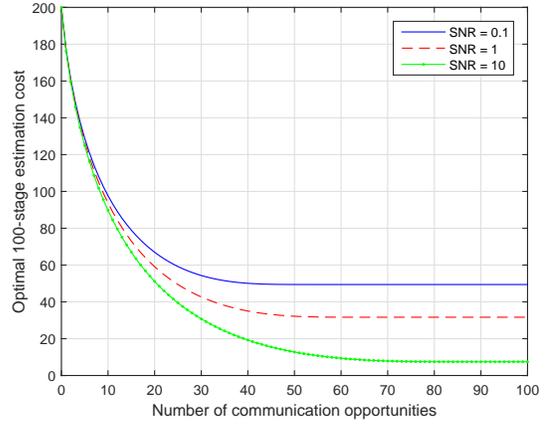}
\caption{100-stage estimation error vs. number of communication opportunities}
\label{fig: 100 stage versus numb of comm}
\end{figure}

One can see that, as to be expected, for each fixed SNR, the optimal 100-stage estimation error is non-increasing in terms of the number of communication opportunities. To be more specific, there exists a threshold on the number of communication opportunities (call it \textit{opportunity threshold}) such that the optimal 100-stage estimation error decreases when the number of communication opportunities is below the threshold, and it stays constant above the threshold. We call \textit{minimal error} as the optimal 100-stage estimation error with the number of communication opportunities above the opportunity threshold. One can also see from Fig. \ref{fig: 100 stage versus numb of comm} that when the SNR increases, the opportunity threshold increases, and the minimal error decreases.

The existence of opportunity threshold can be interpreted as follows: since we restrict the sensor to apply the threshold based policy with threshold $\beta^\ast_t(E_t) = \sqrt{c_t(E_t)+m} \geq \sqrt{m}$, the expected number of communication opportunities that will be used is upper bounded by $T\cdot \Prob (|X_t| \geq \sqrt{m}) = Te^{-\lambda m} $. Therefore when the communication opportunities is greater than $Te^{-\lambda m}$, the optimal expected estimation error will not decrease even though the sensor can have more communication opportunities. It can also be checked from Fig. \ref{fig: 100 stage versus numb of comm} that the opportunity thresholds under different signal to noise ratios are roughly $Te^{-\lambda m}$. Moreover, since $m=\frac{1}{\gamma+1}\frac{1}{\lambda^2}$, $Te^{-\lambda m}=Te^{-\frac{1}{\lambda(\gamma+1)}}$, which is an increasing function of the SNR $\gamma$. Therefore the opportunity threshold increases as the SNR increases.

Fig. \ref{fig:sample path} depicts a sample path of the number of communication opportunities left when the sensor applies the threshold based scheduling policy described in Theorem \ref{Optimal scheduling encoding decoding}. When generating the plot we chose $T=100$, $\lambda=1$, $\gamma = 0.1$, and the number of communication opportunities $N=50$. One can see that all the communication opportunities are not used up by the end of the time horizon.
\begin{figure}[h]
\centering
\includegraphics[height=60mm]{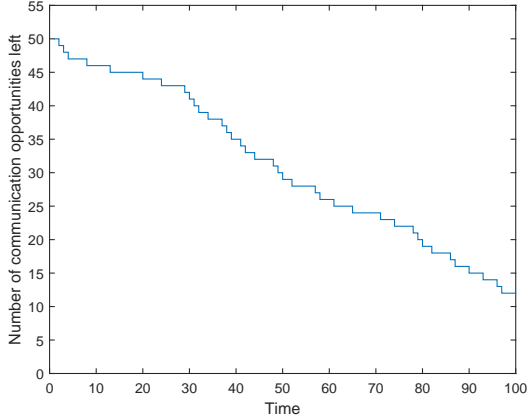}
\caption{A sample path of the number of communication opportunities left vs. time}
\label{fig:sample path}
\end{figure}

When the number of communication opportunities is larger than the opportunity threshold, the optimal estimation error does not change with respect to the number of communication opportunities. Without loss of generality we can assume that the sensor is allowed to communicate at each step, that is, $N=T$. Then the opportunity cost is $c_t(E_t)=0$. Recall that $\beta^\ast_t(E_t) = \sqrt{c_t(E_t)+m}$ and $m=\frac{1}{\gamma+1}\frac{1}{\lambda^2}$. Hence the update rule for the cost function can be simplified as follows:
$$
J^\ast(t,T) = J^\ast(t+1,T) + \Big(\frac{2}{\lambda^2} - \big (\frac{2\sqrt{m}}{\lambda}+\frac{2}{\lambda^2}\big ) \cdot e^{-\lambda\sqrt{m}} \Big)
$$
with $J^\ast(T+1,T)=0$, which implies that
$$
\begin{array}{lcl}
J^\ast(1,T) &=& T \cdot \Big(\dfrac{2}{\lambda^2} - \big (\dfrac{2\sqrt{m}}{\lambda}+\dfrac{2}{\lambda^2}\big ) \cdot e^{-\lambda\sqrt{m}} \Big) \vspace{0.2 cm}
\\ &=& T \cdot 2\lambda^{-2} \cdot \Big[1 - \big (\dfrac{1}{\sqrt{1+\gamma}}+1\big ) \cdot e^{-\frac{1}{\sqrt{1+\gamma}}} \Big]
\end{array}
$$
It is straightforward to check that $J^\ast(1,T)$ is a decreasing function of the SNR $\gamma$. Hence, the minimal error decreases as the SNR increases.

We plot the opportunity threshold $Te^{-\lambda m}$ versus the minimal error $J^\ast(1,T)$ under different SNRs (dash line) in Fig. \ref{fig: 100 stage versus numb of comm} and arrive at Fig. \ref{fig:dash line}.
\begin{figure}[h]
\centering
\includegraphics[height=60mm]{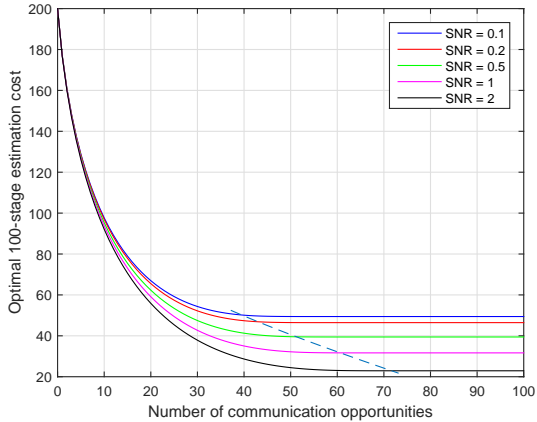}
\caption{Opportunity threshold vs. minimal error under different signal to noise ratios}
\label{fig:dash line}
\end{figure}

One can see that the intersection between the dash line and each solid line is roughly the turning point of the solid line. Therefore, the plot of opportunity threshold versus minimal error under different SNRs is an important one. In fact, the plot suggests the lowest capacity of the battery that one should choose when building a physical system so that the expected estimation error is minimized. Furthermore the plot predicts the minimal expected estimation error.

Consider the asymptotic case where the SNR $\gamma\rightarrow\infty$, and thus $m=\frac{1}{\gamma+1}\frac{1}{\lambda^2}\rightarrow 0$. Then the opportunity threshold $Te^{-\lambda m}\rightarrow T$, and the minimal error $J^\ast(1,T)\rightarrow 0$. Hence, the optimal 100-stage estimation error will be strictly decreasing in terms of the number of communication opportunities in the asymptotic case, as also noted in prior work \cite{Imer10}. Furthermore, the estimation error will reach zero when the number of communication opportunities equals the time horizon.

\section{Conclusions}
\label{Conclusions}
In this paper, we have considered a sensor scheduling and remote estimation problem with limited communication opportunities and noisy communication channel. 
The main contribution is that, as opposed to prior work that assume a perfect communication channel in the problem formulation, we have solved here the problem with an additive noise channel. For a Laplacian source and Gamma channel noise, we have obtained the optimal encoding and estimation policies and  the optimal transmission schedule using dynamic programming.  Our analysis has uncovered a rather surprising result: There might be cases where the sensor does not use all the available transmission opportunities, which is in a sharp contrast with the noiseless setting, where all communication opportunities are used. Future directions for research include extensions to higher dimensional spaces and multi-channel settings.

\bibliographystyle{unsrt}
\bibliography{references}

\begin{thebibliography}{10}

\bibitem{athans1972determination}
M.~Athans.
\newblock On the determination of optimal costly measurement strategies for
  linear stochastic systems.
\newblock {\em Automatica}, 8(4):397--412, 1972.

\bibitem{mehra1976optimization}
R.~Mehra.
\newblock Optimization of measurement schedules and sensor designs for linear
  dynamic systems.
\newblock {\em IEEE Transactions on Automatic Control}, 21(1):55--64, 1976.

\bibitem{Hespanha07}
J.~P. Hespanha, P.~Naghshtabrizi, and Y.~Xu.
\newblock A survey of recent results in networked control systems.
\newblock {\em Proceedings of the IEEE}, 95(1):138--162, 2007.

\bibitem{Imer10}
O.~C. Imer and T.~Ba\c{s}ar.
\newblock Optimal estimation with limited measurements.
\newblock {\em International Journal of Systems Control and Communications},
  2(1-3):5--29, 2010.

\bibitem{Lipsa11}
G.~M. Lipsa and N.~C. Martins.
\newblock Remote state estimation with communication costs for first-order
  {LTI} systems.
\newblock {\em IEEE Transactions on Automatic Control}, 56(9):2013--2025, 2011.

\bibitem{Nayyar13}
A.~Nayyar, T.~Ba\c{s}ar, D.~Teneketzis, and V.~V. Veeravalli.
\newblock Optimal stategies for communication and remote estimation with an
  energy harvesting sensor.
\newblock {\em IEEE Transactions on Automatic Control}, 58(9):2246--2260, 2013.

\bibitem{Wu13}
J.~Wu, Q.~Jia, K.~H. Johansson, and L.~Shi.
\newblock Event-based sensor data scheduling: Trade-off between communication
  rate and estimation quality.
\newblock {\em IEEE Transactions on Automatic Control}, 58(4):1041--1046, 2013.

\bibitem{CoverBook}
T.~Cover and J.~Thomas.
\newblock {\em Elements of Information Theory}.
\newblock John Wiley \& Sons, 2012.

\bibitem{Akyol13}
E.~Akyol, K.~B. Viswanatha, K.~Rose, and T.~A. Ramstad.
\newblock On zero-delay source-channel coding.
\newblock {\em IEEE Transactions on Information Theory}, 60(12):7473--7489,
  2014.

\bibitem{akyol2013optimal}
E.~Akyol, K.~Rose, and T.~Ba\c{s}ar.
\newblock Optimal zero-delay jamming over an additive noise channel.
\newblock {\em IEEE Transactions on Information Theory}, 61(8):4331--4344,
  2015.

\bibitem{Gao15}
X.~Gao, E.~Akyol, and T.~Ba\c{s}ar.
\newblock Optimal sensor scheduling and remote estimation over an additive
  noise channel.
\newblock In {\em Proceedings of American Control Conference (ACC)}, pages
  2723--2728, 2015.

\bibitem{bertsekas1995dynamic}
D.~P Bertsekas.
\newblock {\em Dynamic {P}rogramming and {O}ptimal {C}ontrol, 3rd edition},
  volume~1.
\newblock Athena Scientific Belmont, MA, 2005.

\end{thebibliography}

\end{document}